\documentclass[conference]{IEEEtran}
\ifCLASSINFOpdf
  \usepackage[pdftex]{graphicx}
\else
\fi
\usepackage{amsmath,amssymb,amsfonts,mathrsfs}
\usepackage{balance}
\usepackage{setspace}
\usepackage{cite}
\usepackage{amsthm}
\usepackage{graphicx}
\usepackage{epstopdf}
\usepackage{caption}
\usepackage{algorithm}

\usepackage{bbm}

\usepackage{amsfonts}
\usepackage{bm}
\usepackage{geometry}
\usepackage{array}
\IEEEoverridecommandlockouts

\begin{document}
\title{Deterministic Scheduling for Low-latency Wireless Transmissions with Continuous Channel States}

\author{\IEEEauthorblockN{Junjie Wu,\textit{ Student Member,~IEEE}, Wei Chen,\textit{ Senior Member,~IEEE}}
\IEEEauthorblockA{Department of Electronic Engineering, Tsinghua University, Beijing, 100084, CHINA\\
Beijing National Research Center for Information Science and Technology\\
Email: wujj18@mails.tsinghua.edu.cn, wchen@tsinghua.edu.cn}	\thanks{This research is supported by the National Program for Special Support for Eminent Professionals of China (National 10,000-Talent Program), the Beijing Natural Science Foundation under Grant No. 4191001, and the National Science Foundation of China under Grant No. 61671269.} }


%Telephone: (800) 555--1212, Fax: (888) 555--1212}


\maketitle

\begin{abstract}
  High energy efficiency and low latency have always been the significant goals pursued by the designer of wireless networks. One efficient way to achieve these goals is cross-layer scheduling based on the system states in different layers, such as queuing state and channel state. However, most existing works in cross-layer design focus on the scheduling based on the discrete channel state. Little attention is paid to considering the continuous channel state that is closer to the practical scenario. Therefore, in this paper, we study the optimal cross-layer scheduling policy on data transmission in a single communication link with continuous state-space of channel. The aim of scheduling is to minimize the average power for data transmission when the average delay is constrained. More specifically, the optimal cross-layer scheduling problem was formulated as a variational problem. Based on the variational problem, we show the optimality of the deterministic scheduling policy through a constructive proof. The optimality of the deterministic policy allows us to compress the searching space from the probabilistic policies to the deterministic policies.
\end{abstract}

% \fi
\IEEEpeerreviewmaketitle

\section{Introduction}
With the growth of the Internet of Things (IoT), billions of devices are connected through wireless networks, which puts great challenges in satisfying the strict quality of service (QoS). Among the strict QoS, high energy efficiency and ultra-low latency are the important goals in the present 5G and future 6G [1, 2]. In particular, to achieve high energy efficiency and ultra-low latency simultaneously, the cross-layer design is a significant way. To this end, in our work, we focus on minimizing average power for data transmission with the constraint of average delay through the cross-layer design.

In existing works [3-9], the cross-layer design has been widely studied and implemented as a significant way to satisfy the strict QoS through scheduling and resource allocation. More specifically, works [3-6] focused on the optimal scheduling of data transmission under the cross-layer framework. In particular, the Collins and Cruz proposed the pioneering work [3] on cross-layer design to minimize average power for data transmission over a two-state fading channel. In [4], through cross-layer design, Berry and Gallager studied the optimal delay-power tradeoff in the regime of asymptotically large delay. After this, Berry also studied the optimal delay-power tradeoff in the regime of asymptotically small delay in [5]. Authors in [6] proposed a lazy scheduling algorithm to achieve energy-efficient transmission. There are also some works focused on the resource allocation and management with cross-layer design [7-9]. In [7], the power and bandwidth allocation policy were studied in the scenario of Ultra-Reliable and Low-Latency (URLLC) to minimize the power for data transmission under the QoS constraint. The optimal energy allocation was studied in [8] to maximize the throughput with the consideration of energy harvesting. In [9], the energy allocation and management were studied in satellites communications through a dynamic programming approach.

In our previous work [10-12], we studied the optimal delay-power tradeoff and its corresponding scheduling policy through the cross-layer design. More particularly, we first proposed the probabilistic scheduling to minimize average delay under the constraint of average power in a simple scenario with the two-state block-fading channel and fixed modulation [10]. Based on the probabilistic scheduling proposed in [10], in the follow-up works, we considered the optimal delay-power tradeoff in more generalization and complex scenarios. We considered  buffer aware scheduling with adaptive transmission in [11] and multi-state block-fading channel in [12]. However, in the practical scenario, the channel energy gain fluctuates in a continuous range. If we model the state of channel as multiple discrete states, the practical channel state information will not be accurately described.

To make full use of channel state information, in this paper, we study the optimal cross-layer scheduling based on the continuous channel states. Fortunately, the optimal scheduling can be converted to a variational problem. Based on this, we construct a solution within the feasible region of the variational problem. The scheduling policy corresponding to the constructed solution is shown to be the deterministic policy. Meanwhile, the average power corresponding to the constructed solution is shown to approximate the achievable minimum power with an arbitrarily small error. Therefore, based on the constructed solution, we finally show the optimality of deterministic cross-layer scheduling policy. This conclusion greatly compresses the search space for finding the optimal policy.

\section{System Model}

Consider a communication link that a transmitter sends data packets to its receiver over a block fading channel as shown in Fig. 1. The data transmission is scheduled according to the information of the discrete data queue state and continuous channel state under the cross-layer design. The aim of scheduling is to minimize average power for data transmission with the constraint of average delay.

In the network layer, data packets arrive at the transmitter randomly at the beginning of each time slot. Let $a[n]\in \mathcal{A}$ denote the number of data packets arriving in time slot $n$, where $\mathcal{A}=\{0, 1, 2, ..., A\}$. In particular, the parameter $A$ is the maximum number of data packet that can be arrived in one time slot. To capture the practical scenario, we allow that $a[n]$  could follow any discrete probability distribution. Therefore, the probability distribution of $a[n]$ can be given by
 \begin{equation}
\Pr\{a[n]=k\}=\alpha_{k},
 \end{equation}
 where $\alpha_{k}\geq0$ and $\sum_{k=0}^{A}\alpha_{k}=1$. Thus, the average arrival rate of data packets $\overline{a}$ can be obtained as
 \begin{equation}
\overline{a}=\sum_{k=0}^{A}k\alpha_{k}.
 \end{equation}

 In the data link layer, the arrived data packets will be backlogged in the data queue if they are not sent out immediately. We let $Q$ denote the buffer size of data queue. It means that if the number of packets in the data queue exceeds $Q$, the extra packets will be lost. Let $q[n]\in \mathcal{Q}$ denote the data queue length at the beginning of time slot $n$, where $\mathcal{Q}$ is the discrete space that can be expressed as $\{ 1, 2, ..., Q\}$. We implement $q[n]$ to indicate the state of data queue for scheduling of data transmission. The update of $q[n]$ is given by
  \begin{equation}
   q[n+1]=\min\{\max\{q[n]-s[n], 0\}+a[n+1],Q\},
 \end{equation}
 where $s[n]\in\mathcal{S}$ is the number of packets transmitted in time slot $n$, and $\mathcal{S}=\{1, 2, ..., S_{max}\}$. The $S_{max}$ indicates the maximum number of data packets that can be transmitted in one time slot corresponding to the transmitter's capacity for data transmission. We assume that $S_{max}\geq A$ to avoid packet loss caused by overflow of data queue.

In the physical layer, the data packets are transmitted over the block fading channel. We assume that channel energy gain stays invariant during each time slot, and follows an independent and identically distributed ($i.i.d.$) fading process across the time slots. In practice, the channel energy gain is a value that fluctuates in a continuous interval. Let $h[n]$ denote the channel energy gain in time slot $n$. To capture the practical scenario without loss of generality, we assume that $h[n]\in \mathcal{H}$ obeys a probability density function $f_{H}(h)$, where $\mathcal{H}=(h_{min}, h_{max}]$, and $f_{H}(h)$ is a continuous bounded function on $\mathcal{H}$. Therefore, the probability distribution of $h[n]$ can be given by
 \begin{equation}
\Pr\{h[n]\leq x\}=\int_{h_{min}}^{x} f_{H}(h)  dh.
 \end{equation}
  \begin{figure}
   \centering
   \includegraphics[width=0.5\textwidth]{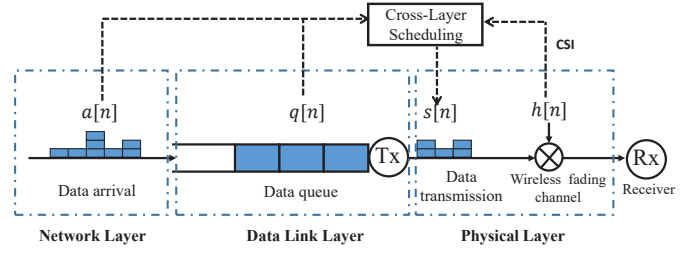}
   \caption{System Model.}
   \label{fig.pathdemo4}
   \end{figure}
The energy consumption for data transmission in time slot $n$ is determined by the data transmission rate $s[n]$ and channel energy gain $h[n]$. Let $\xi(s)h^{-1}$ denote the energy consumption in one time slot when transmission rate is $s$ and the channel energy gain is $h$. The $\xi(s)$ is assumed to be a strictly monotonic increasing function, since more data transmission requires more energy. For example, $\xi(s)=2^{s}-1$ is widely used, which can be derived based on channel capacity.

Our aim is to minimize the average power for data transmission while satisfying that the average delay is below a certain threshold $D_{th}$.

\section{The Deterministic Policy is Optimal}
In this section, we first present the probabilistic scheduling that can cover all feasible stationary policies for data transmission. Then the minimization of average power for data transmission can be formulated as a variational problem. Based on this variational problem, we show the optimality of deterministic policy through a constructive proof.

The probabilistic scheduling can be expressed by a set of scheduling parameters $\{f(s|q,h)\}$. The parameter $f(s|q,h)$ denotes the probability of transmitting $s$ data packets when the data queue length is $q$ and channel energy gain is $h$. Let $g(q, s, h)$ represent the steady state probability density of the system being in a state that the data queue length is $ q $, the transmission rate is $ s$, and the channel energy gain is $h$.

Based on the definition of $g(q, s, h)$, the scheduling parameters $\{f(s|q,h)\}$ can be represented by $g (q, s, h)$ as follows,
 \begin{equation}
 f(s|q,h)=\frac{g (q, s, h)}{\sum_{x=0}^{S_{max}}g (q, x, h)}.
 \end{equation}
 Furthermore, the average delay $D_{ave}$ and average power $P_{ave}$ can also be represented by $g (q, s, h)$. The average delay is associated with average data queue length. The average data queue length can be obtained as
 \begin{equation}
\mathbb{E}\{q[n]\}=\int_{h_{min}}^{h_{max}}\sum_{s=0}^{S_{max}}\sum_{q=0}^{Q}q g(q, s, h)dh.
\end{equation}
According to Little's Law, the expression of $D_{ave}$ can be obtained as
 \begin{equation}
D_{ave}=\frac{1}{\overline{a}}\int_{h_{min}}^{h_{max}}\sum_{s=0}^{S_{max}}\sum_{q=0}^{Q}qg(q, s, h)dh.
\end{equation}
The average power for data transmission is associated with transmission rate and channel energy gain. Thus, the expression of $P_{ave}$ can be obtained as
 \begin{equation}
 P_{ave}=\int_{h_{min}}^{h_{max}}\sum_{s=0}^{S_{max}}\sum_{q=0}^{Q}\xi(s)h^{-1}g(q, s, h)dh.
\end{equation}

In Eqs. (5), (7), and (8), the probabilistic scheduling parameters $\{f(s|q,h)\}$, the average delay $D_{ave}$, and the average power $P_{ave}$ are all represented by steady state probability density function $g (q, s, h)$, respectively. This allows us to use $g (q, s, h)$ as the optimization variable to minimize the average power as follows,
\begin{subequations} \label{10151408}
\renewcommand{\theequation}{\theparentequation.\alph{equation}}
\begin{align}
\min_{g(q,s,h)}\ & ~  \int_{h_{min}}^{h_{max}}\sum_{s=0}^{S_{max}}\sum_{q=0}^{Q}\xi(s)h^{-1}g(q, s, h)dh\\
 s.t.\ \ &~ \frac{1}{\overline{a}}\int_{h_{min}}^{h_{max}}\sum_{s=0}^{S_{max}}\sum_{q=0}^{Q}qg(q, s, h)dh\leq D_{th} \label{b} \\
& ~\sum_{s=0}^{S_{max}}\sum_{q=0}^{Q}g(q, s, h)=f_{H}(h)\ \ \ \forall h\in \mathcal{H}\label{c}\\
& ~\sum_{k=0}^{A}\sum_{s=0}^{S_{max}}\alpha_{k}\int_{h_{min}}^{h_{max}}g(q+s-k, s, h)dh=\notag\\ & ~ \sum_{s=0}^{S_{max}}\int_{h_{min}}^{h_{max}}g(q, s, h)dh\ \ \ \forall q\in \mathcal{Q}\label{c}\\
& ~ g(q,s,h)\geq 0 \ \ \forall q\in \mathcal{Q}, s\in \mathcal{S}, h\in \mathcal{H}\\
& ~ g(q,s,h)=0 \ \ \ \forall q\!-\!s<0\notag \\ & ~  \ \text{or}\ q\!-\!s>Q\!-\!A \ \text{or}\ q \notin \mathcal{Q}\ \text{or}\ s \notin \mathcal{S}.
\end{align}%
\end{subequations}

In problem (9), the objective function (9.a) is the average power, as shown in Eq. (8). The first constraint (9.b) is the constraint of average delay, as shown in Eq. (7). The second constraint (9.c) is the constraint that channel energy gain obeys the probability density function $f_{H}(h)$. The constraint (9.d) represents the state balance equation of the data queue state. The left side of the equation represents the summation of probability that the data queue state are transferred to the queue state $q$ at the next time slot, while the right side of the equation is the steady-state probability that the system is in the queue state $q$. Constraint (9.e) is the non-negative condition of probability. Finally, constraint (9.f) avoids the overflow or underflow of the data queue. However, problem (9) is a variational problem, since its optimization variable $g (q, s, h)$ is a probability density function.

Though, it is not a trivial work to solve problem (9), we may show a light on the structural properties of its optimal solution. Let $g ^{*}(q, s, h)$ denote the optimal solution of the problem (9). Based on $g ^{*}(q, s, h)$, we define a set of functions $\{U_{q}(x)\}$ as follows,
 \begin{equation}
U_{q}(x)=\int_{h_{min}}^{x}\sum_{s=0}^{S_{max}}g ^{*}(q, s, h)dh,
 \end{equation}
where $q\in\mathcal{Q}$. Then, we construct a solution $y_{M}(q, s, h)$ as follows
 \begin{equation}
y_{M}(q, s, h)=\sum_{k=1}^{M}\sum_{x=0}^{S_{max}}g ^{*}(q, x, h)\mathbb{I}\{h\in(h_{k,s,q}^{min}, h_{k,s,q}^{max}]\},
 \end{equation}
where $M\in\mathbb{N}$, $q\in\mathcal{Q}$, $h\in \mathcal{H}$. The definition of $h_{k,s,q}^{max}$ in Eq. (11) is given by
 \begin{equation}
 \begin{aligned}
 & ~h_{k,s,q}^{max}=U_{q}^{-1}(\int_{h_{min}}^{h_{k}^{min}}\sum_{x=0}^{S_{max}}g ^{*}(q, x, h)dh \\ & ~+\int_{h_{k}^{min}}^{h_{k}^{max}}\sum_{x=0}^{s}g ^{*}(q, x, h)dh),
  \end{aligned}
 \end{equation}
 where $U_{q}^{-1}(.)$ is the inverse of $U_{q}(.)$, and $h_{k}^{min}=h_{min}+\frac{(k-1)(h_{max}-h_{min})}{M}$, $h_{k}^{max}=h_{min}+\frac{k(h_{max}-h_{min})}{M}$. We also define $h_{k,-1,q}^{max}=h_{k}^{min}$. Then, the definition of $h_{k,s,q}^{min}$ in Eq. (11) is given by
 \begin{equation}
h_{k,s,q}^{min}=h_{k,s-1,q}^{max}, \ \ \ \forall \ \ 0\leq s \leq S_{max}.
 \end{equation}

 For the constructed solution $y_{M}(q, s, h)$ in Eq. (11), we first show that $y_{M}(q, s, h)$ is in the feasible region of problem (9). On this basis, we show that the scheduling policy corresponding to $y_{M}(q, s, h)$ is the deterministic policy. Finally, we show that the average power corresponding to $y_{M}(q, s, h)$ can approximate that of the optimal solution $g ^{*}(q, s, h)$ with an arbitrarily small error for the sufficiently large $M$. To sum up, the optimality of the deterministic scheduling policy can be shown.
   \newtheorem{lem}{\bf Lemma}
   \begin{lem}\label{lem1}
   The constructed solution $y_{M}(q, s, h)$ in Eq. (11) satisfies the constraint (9.c), i.e.,
     \begin{equation}
    \sum_{s=0}^{S_{max}}\sum_{q=0}^{Q}y_{M}(q, s, h)=f_{H}(h)\ \ \ \forall h\in \mathcal{H}.
     \end{equation}
     \end{lem}
 \begin{proof}
 Based on Eq. (12), we obtain the $h_{k,S_{max},q}^{max}$ given by
  \begin{equation}
 \begin{aligned}
  h_{k,S_{max},q}^{max}&=U_{q}^{-1}(\int_{h_{min}}^{h_{k}^{min}}\sum_{x=0}^{S_{max}}g ^{*}(q, x, h)dh \\ & ~+\int_{h_{k}^{min}}^{h_{k}^{max}}\sum_{x=0}^{S_{max}}g ^{*}(q, x, h)dh)\\
              &=U_{q}^{-1}(\int_{h_{min}}^{h_{k}^{max}}\sum_{x=0}^{S_{max}}g ^{*}(q, x, h)dh)\\
              &=h_{k}^{max}.
   \end{aligned}
   \end{equation}

   Therefore, we have
    \begin{equation}
 \begin{aligned}
  \bigcup_{k=1}^{M}\bigcup_{s=0}^{S_{max}} (h_{k,s,q}^{min}, h_{k,s,q}^{max}]&=\bigcup_{k=1}^{M} (h_{k,0,q}^{min}, h_{k,S_{max},q}^{max}]\\
              &=\bigcup_{k=1}^{M} (h_{k}^{min}, h_{k}^{max}]\\
              &=(h_{min}, h_{max}].
   \end{aligned}
   \end{equation}
 According to Eqs. (11) and (16), we can simply $\sum_{s=0}^{S_{max}}\sum_{q=0}^{Q}y_{M}(q, s, h)$ as follows,
      \begin{equation}
 \begin{aligned}
 \sum_{s=0}^{S_{max}}\sum_{q=0}^{Q}y_{M}(q, s, h)&=\sum_{q=0}^{Q}\sum_{x=0}^{S_{max}}\sum_{s=0}^{S_{max}}\sum_{k=1}^{M}\\ & ~g ^{*}(q, x, h)\mathbb{I}\{h\in(h_{k,s,q}^{min}, h_{k,s,q}^{max}]\}\\
              &=\!\!\sum_{q=0}^{Q}\!\!\sum_{x=0}^{S_{max}}\!\!g ^{*}(q, x, h)\mathbb{I}\{h\!\!\in\!\!(h_{min}, h_{max}]\}\\
              &=\sum_{q=0}^{Q}\!\!\sum_{x=0}^{S_{max}}\!\!g ^{*}(q, x, h)\\
              &=f_{H}(h).
   \end{aligned}
   \end{equation}

 \end{proof}

  \begin{lem}\label{lem2}
   The constructed solution $y_{M}(q, s, h)$ in Eq. (11) satisfies the constraint (9.d), i.e.,
  \begin{equation}
 \begin{aligned}
& ~\sum_{k=0}^{A}\sum_{s=0}^{S_{max}}\alpha_{k}\int_{h_{min}}^{h_{max}}y_{M}(q+s-k, s, h)dh=\\ & ~ \sum_{s=0}^{S_{max}}\int_{h_{min}}^{h_{max}}y_{M}(q, s, h)dh\ \ \ \forall q\in \mathcal{Q}.
   \end{aligned}
   \end{equation}
     \end{lem}
 \begin{proof}
 We first show that for any $q\in \mathcal{Q}$, $s\in \mathcal{S}$, we have
 \begin{subequations} \label{10151408}
\renewcommand{\theequation}{\theparentequation.\alph{equation}}
 \begin{align}
& ~\int_{h_{min}}^{h_{max}}y_{M}(q, s, h)dh\\
 & ~ =\int_{h_{min}}^{h_{max}}\!\!\sum_{k=1}^{M}\sum_{x=0}^{S_{max}}\!\!g ^{*}(q, x, h)\mathbb{I}\{h\!\!\in\!\!(h_{k,s,q}^{min}, h_{k,s,q}^{max}]\}dh\\
  & ~ =\sum_{k=1}^{M}\int_{h_{k,s,q}^{min}}^{h_{k,s,q}^{max}}\sum_{x=0}^{S_{max}}g ^{*}(q, x, h)dh\\
    & ~ =\sum_{k=1}^{M}(U_{q}(h_{k,s,q}^{max})-U_{q}(h_{k,s,q}^{min}))\\
      & ~ =\sum_{k=1}^{M}(U_{q}(h_{k,s,q}^{max})-U_{q}(h_{k,s-1,q}^{max}))\\
      & ~ =\sum_{k=1}^{M}((\!\!\int_{h_{min}}^{h_{k}^{min}}\sum_{x=0}^{S_{max}}\!\!g ^{*}(q, x, h)dh \!\!+\!\!\!\int_{h_{k}^{min}}^{h_{k}^{max}}\!\!\sum_{x=0}^{s}\!\!g ^{*}(q, x, h)dh)\notag\\
      & ~-\!\!(\!\!\int_{h_{min}}^{h_{k}^{min}}\sum_{x=0}^{S_{max}}\!\!g ^{*}(q, x, h)dh \!\!+\!\!\!\int_{h_{k}^{min}}^{h_{k}^{max}}\sum_{x=0}^{s-1}\!\!g ^{*}(q, x, h)dh))\\
          & ~ =\sum_{k=1}^{M}\int_{h_{k}^{min}}^{h_{k}^{max}}g ^{*}(q, s, h)dh\\
          & ~ =\int_{h_{min}}^{h_{max}}g ^{*}(q, s, h)dh.
   \end{align}
 \end{subequations}

We obtain (19.b) through substituting $y_{M}(q, s, h)$ in (19.a) which is based on Eq. (11). By changing the integration interval, we can eliminate the indicative function in (19.b), as shown in (19.c). The definition of $U_{q}(x)$ in Eq. (10) helps us convert (19.d) to (19.c). We obtain (19.e) through substituting $h_{k,s,q}^{min}$ in (19.a) by $h_{k,s-1,q}^{max}$ which is based on Eq. (13). Based on the definition of $h_{k,s,q}^{max}$ in Eq. (12), we obtain (19.f) from (19.e). Finally, (19.h) is derived.

The optimal solution $g ^{*}(q, s, h)$ satisfies the constraint (9.d). Therefore, based on Eq. (19), we have
  \begin{equation}
 \begin{aligned}
& ~\sum_{k=0}^{A}\sum_{s=0}^{S_{max}}\alpha_{k}\int_{h_{min}}^{h_{max}}y_{M}(q+s-k, s, h)dh\\& ~ = \sum_{k=0}^{A}\sum_{s=0}^{S_{max}}\alpha_{k}\int_{h_{min}}^{h_{max}}g ^{*}(q+s-k, s, h)dh\\& ~
=\sum_{s=0}^{S_{max}}\int_{h_{min}}^{h_{max}}g ^{*}(q, s, h)dh\\& ~
=\sum_{s=0}^{S_{max}}\int_{h_{min}}^{h_{max}}y_{M}(q, s, h)dh,
   \end{aligned}
   \end{equation}
where $ q\in \mathcal{Q}$.

  \end{proof}

  \begin{lem}\label{lem2}
   The constructed solution $y_{M}(q, s, h)$ in Eq. (11) satisfies the constraint (9.b), i.e.,
  \begin{equation}
 \frac{1}{\overline{a}}\int_{h_{min}}^{h_{max}}\sum_{s=0}^{S_{max}}\sum_{q=0}^{Q}qy_{M}(q, s, h)dh\leq D_{th}.
   \end{equation}
     \end{lem}
\begin{proof}
   Based on Eq. (19), we have
  \begin{equation}
 \begin{aligned}
& ~\frac{1}{\overline{a}}\int_{h_{min}}^{h_{max}}\sum_{s=0}^{S_{max}}\sum_{q=0}^{Q}qy_{M}(q, s, h)dh\\ & ~ =\frac{1}{\overline{a}}\int_{h_{min}}^{h_{max}}\sum_{s=0}^{S_{max}}\sum_{q=0}^{Q}qg^{*}(q, s, h)dh\\ & ~
\leq D_{th}.
   \end{aligned}
   \end{equation}
\end{proof}
  \begin{lem}\label{lem2}
   The constructed solution $y_{M}(q, s, h)$ in Eq. (11) satisfies the constraint (9.e), i.e.,
  \begin{equation}
y_{M}(q,s,h)\geq 0 \ \ \forall q\in \mathcal{Q}, s\in \mathcal{S}, h\in \mathcal{H}.
   \end{equation}
     \end{lem}
\begin{proof}
Since $g^{*}(q, s, h)$ satisfies the constraint (9.e) in problem (9), based on Eq. (11) $y_{M}(q, s, h)$ also satisfies the constraint (9.e).
\end{proof}

  \begin{lem}\label{lem2}
   The constructed solution $y_{M}(q, s, h)$ satisfies the constraint (9.f) in problem (9), i.e.,
  \begin{equation}
  y_{M}(q,s,h)=0 \ \ \ \forall q\!-\!s<0 \ \text{or}\ q\!-\!s>Q\!-\!A \ \text{or}\ q \notin \mathcal{Q}\ \text{or}\ s \notin \mathcal{S}.
   \end{equation}
     \end{lem}
\begin{proof}
Since $g^{*}(q, s, h)$ satisfies the constraint (9.f), if $q-s<0\ \text{or}\ q-s>Q-A\ \text{or}\ q \notin \mathcal{Q}\ \text{or}\ s \notin \mathcal{S}$, then we have
  \begin{equation}
  g^{*}(q, s, h)=0,
   \end{equation}
 When $g^{*}(q, s, h)=0$, we have
   \begin{equation}
 \begin{aligned}
h_{k,s,q}^{max}-h_{k,s-1,q}^{max}& ~ =\int_{h_{k}^{min}}^{h_{k}^{max}}g ^{*}(q, s, h)dh\\ & ~
=0.
   \end{aligned}
   \end{equation}
 Based on Eqs. (13) and (26), we have
 \begin{equation}
h_{k,s,q}^{max}=h_{k,s,q}^{min}.
 \end{equation}
 Finally, based on Eq. (11), using the relationship $h_{k,s,q}^{max}=h_{k,s,q}^{min}$ gives $y_{M}(q, s, h)=0$.
\end{proof}

   \newtheorem{thm}{\bf Theorem}
   \begin{thm}\label{thm1}
   The constructed solution $y_{M}(q, s, h)$ is in the feasible region of problem (9).
\end{thm}
\begin{proof}
Based on Lemma 1-5, Theorem 1 is proved.
\end{proof}
 \begin{lem}\label{lem2}
   For any $q\in\mathcal{Q}$, $h\in\mathcal{H}$, there is at most one $s\in\mathcal{S}$ satisfies that

    \begin{equation}
    y_{M}(q, s, h)>0.
    \end{equation}

  \end{lem}
 \begin{proof}
 Based on Eqs. (13) (15), for any $k_{1}, k_{2}\in [1:M]$, $s_{1}, s_{2}\in \mathcal{S}$ and $q\in\mathcal{Q}$, we have $(h_{k_{1},s_{1},q}^{min}, h_{k_{1},s_{1},q}^{max}]\in (h_{k_{1}}^{min}, h_{k_{1}}^{max}]$ and $(h_{k_{2},s_{2},q}^{min}, h_{k_{2},s_{2},q}^{max}]\in (h_{k_{2}}^{min}, h_{k_{2}}^{max}]$.

 If $k_{1}\neq k_{2}$, then based on the definition of $h_{k}^{min}$ and $h_{k}^{max}$ we have
     \begin{equation}
  (h_{k_{1},s_{1},q}^{min}, h_{k_{1},s_{1},q}^{max}]\bigcap (h_{k_{2},s_{2},q}^{min}, h_{k_{2},s_{2},q}^{max}]=\emptyset.
   \end{equation}
  If $k_{1}= k_{2}$ and $s_{1}\neq s_{2}$, then based on Eqs. (12) (13), we still  have
     \begin{equation}
  (h_{k_{1},s_{1},q}^{min}, h_{k_{1},s_{1},q}^{max}]\bigcap (h_{k_{2},s_{2},q}^{min}, h_{k_{2},s_{2},q}^{max}]=\emptyset.
   \end{equation}
   Furthermore, in Eq. (16), we have obtained that
    \begin{equation}
  \bigcup_{k=1}^{M}\bigcup_{s=0}^{S_{max}} (h_{k,s,q}^{min}, h_{k,s,q}^{max}]=(h_{min}, h_{max}].
   \end{equation}
  Therefore, based on Eqs. (29) (30) and (31), for any $h\in(h_{min}, h_{max}]$ and $q\in\mathcal{Q}$, there is only one $k^{*}$ and $s^{*}$ satisfies that
   \begin{equation}
  \mathbb{I}\{h\in(h_{k^{*},s^{*},q}^{min}, h_{k^{*},s^{*},q}^{max}]\}=1.
   \end{equation}
   Finally, based on Eqs. (11) and (32), Lemma 6 is proved.
\end{proof}
Let $\{f_{M}(s|q,h)\}$ denote the scheduling policy corresponding to constructed $y_{M}(q, s, h)$.
\begin{thm}\label{thm1}
    The scheduling policy $\{f_{M}(s|q,h)\}$ is the deterministic policy.
\end{thm}
\begin{proof}
Based on Eq. (5), the scheduling parameters $\{f_{M}(s|q,h)\}$ can be obtained as
 \begin{equation}
 f_{M}(s|q,h)=\frac{y_{M}(q, s, h)}{\sum_{x=0}^{S_{max}}y_{M}(q, x, h)}.
 \end{equation}
 If the denominator in Eq. (33) is $0$, the queue state $q$ is the transient state. If the denominator in Eq. (33) is not $0$. Then, based on Lemma 6, the $f_{M}(s|q,h)$ in Eq. (33) is $0$ or $1$, which means that the scheduling policy corresponding to $y_{M}(q, s, h)$ is the deterministic policy.
\end{proof}
Let $P_{M}$ and $P^{*}$ denote the average power corresponding to the constructed solution $y_{M}(q, s, h)$ and the optimal solution $g^{*}(q, s, h)$, respectively. Next, we explore the relationship between the $P_{M}$ and $P^{*}$.

\begin{lem}\label{lem2}
   The relationship between $P_{M}$ and $P^{*}$ satisfies that
    \begin{equation}
    P^{*}\leq P_{M}\leq (1+\frac{h_{max}-h_{min}}{M h_{min}})P^{*}.
    \end{equation}
  \end{lem}
 \begin{proof}
 Since $g^{*}(q, s, h)$ is the optimal solution of problem (9), then $P^{*}\leq P_{M}$ is established.
  \begin{subequations} \label{10151408}
\renewcommand{\theequation}{\theparentequation.\alph{equation}}
 \begin{align}
P_{M}& ~= \int_{h_{min}}^{h_{max}}\sum_{s=0}^{S_{max}}\sum_{q=0}^{Q}\xi(s)h^{-1}y_{M}(q, s, h)dh\\
 & ~ =\int_{h_{min}}^{h_{max}}\sum_{s=0}^{S_{max}}\sum_{q=0}^{Q}\xi(s)h^{-1}\sum_{k=1}^{M}\sum_{x=0}^{S_{max}}\notag\\  & ~ g ^{*}(q, x, h)\mathbb{I}\{h\in(h_{k,s,q}^{min}, h_{k,s,q}^{max}]\}dh\\
 & ~ =\!\!\sum_{s=0}^{S_{max}}\sum_{k=1}^{M}\sum_{q=0}^{Q}\int_{h_{k,s,q}^{min}}^{h_{k,s,q}^{max}}\xi(s)h^{-1}  \!\!\sum_{x=0}^{S_{max}}\!\!g ^{*}(q, x, h)dh\\
& ~ \leq\!\!\sum_{s=0}^{S_{max}}\sum_{k=1}^{M}\sum_{q=0}^{Q}\frac{\xi(s)}{h_{k}^{min}}\int_{h_{k,s,q}^{min}}^{h_{k,s,q}^{max}}  \!\!\sum_{x=0}^{S_{max}}\!\!g ^{*}(q, x, h)dh\\
 & ~  =\!\!\sum_{s=0}^{S_{max}}\sum_{k=1}^{M}\sum_{q=0}^{Q}\frac{\xi(s)}{h_{k}^{min}}(U_{q}(h_{k,s,q}^{max})-U_{q}(h_{k,s,q}^{min}))\\
& ~  =\!\!\sum_{s=0}^{S_{max}}\sum_{k=1}^{M}\sum_{q=0}^{Q}\frac{\xi(s)}{h_{k}^{min}}(U_{q}(h_{k,s,q}^{max})-U_{q}(h_{k,s-1,q}^{max}))\\
      & ~ =\!\!\sum_{s=0}^{S_{max}}\sum_{k=1}^{M}\sum_{q=0}^{Q}\frac{\xi(s)}{h_{k}^{min}}\int_{h_{k}^{min}}^{h_{k}^{max}}g ^{*}(q, s, h)dh\\
       & ~ \leq\!\!\sum_{s=0}^{S_{max}}\sum_{k=1}^{M}\sum_{q=0}^{Q}\frac{\xi(s)}{h_{k}^{min}}\int_{h_{k}^{min}}^{h_{k}^{max}}\frac{h_{k}^{max}}{h}g ^{*}(q, s, h)dh\\
          & ~ =\!\!\sum_{s=0}^{S_{max}}\sum_{q=0}^{Q}\sum_{k=1}^{M}\frac{h_{k}^{max}}{h_{k}^{min}}\int_{h_{k}^{min}}^{h_{k}^{max}}\!\!\!\xi(s)h^{-1}g ^{*}(q, s, h)dh\\
   & ~ \leq\!\!\sum_{s=0}^{S_{max}}\sum_{q=0}^{Q}\sum_{k=1}^{M}\frac{h_{1}^{max}}{h_{1}^{min}}\int_{h_{k}^{min}}^{h_{k}^{max}}\!\!\!\xi(s)h^{-1}g ^{*}(q, s, h)dh\\
      &~= (1+\frac{h_{max}-h_{min}}{M h_{min}})P^{*}.
   \end{align}
 \end{subequations}

The indicative function in (35.b) is eliminated by changing the integration interval of $h$, which is shown in (35.c). Since $h_{k}^{min}\leq h_{k,s,q}^{min}$, formula (35.c) can be enlarged to (35.d) by replacing $h$ with $h_{k}^{min}$. The definition of $U_{q}(x)$ in Eq. (10) helps us convert formula (35.d) to (35.e). Based on Eq. (13), the $h_{k,s,q}^{min}$ in (35.e) can be replaced by  $h_{k,s-1,q}^{max}$ as shown in (35.f). Fortunately, based on the definition of $h_{k,s,q}^{max}$ in Eq. (12), the $U_{q}(h_{k,s,q}^{max})-U_{q}(h_{k,s-1,q}^{max})$ in (35.f) can be simplified, which is shown in (35.g). Since $\frac{h_{k}^{max}}{h_{k}^{min}}\leq \frac{h_{1}^{max}}{h_{1}^{min}}$ for any $k\in[1:M]$, formula (35.i) can be enlarged to (35.j). Finally, based on Eq. (8), we can obtain (35.k).
  \end{proof}
 \begin{thm}\label{thm1}
    For any $\epsilon>0$, there exists an $m=\frac{h_{max}-h_{min}}{\epsilon h_{min}}$. For all $ M\geq m$, we have
        \begin{equation}
    1\leq \frac{P_{M}}{P^{*}}\leq 1+\epsilon.
    \end{equation}
\end{thm}
 \begin{proof}
 The proof of Theorem 3 is based on Lemma 7. Since  $P_{M}\geq P^{*}$, we can obtain $\frac{P_{M}}{P^{*}}\geq 1$. Based on Eq. (34), we have $\frac{P_{M}}{P^{*}}\leq (1+\frac{h_{max}-h_{min}}{M h_{min}})$.
 Since $M\geq \frac{h_{max}-h_{min}}{\epsilon h_{min}}$, we have $\frac{P_{M}}{P^{*}}\leq 1+\epsilon$.

  \end{proof}
\section{ Numerical Results }
In this section, we adopt Constrained Markov Decision Process (CMDP) to demonstrate the optimality of deterministic policies. More specifically, the continuous state-space of channel $\mathcal{H}$ can be equally divided into multiple subspaces. By judging which subspace the current channel state belongs to, we convert the continuous channel state into the discrete channel state for utilization. Based on this, the optimal scheduling can be approximately converted to CMDP. In the simulations, we set $A=2$, $S_{max}=2$, $\alpha_{0}=0.4$, $\alpha_{1}=0.3$, $\alpha_{2}=0.3$, $Q=10$, $\xi(s)=2^{s}-1$, $\mathcal{H}=(0.5, 10]$, and $ f_{H} (h) $ is a uniform distribution on $\mathcal{H}$.

 In Fig. 2, delay-power tradeoff curves are obtained by CMDP with $\mathcal{H}$ being equally divided into 2, 4, 8, 16 subspaces respectively. As we can see, when $\mathcal{H}$ is more densely divided, the delay-power tradeoff curve becomes better and converges asymptotically. The reason is that when the state-space of channel $\mathcal{H}$ is divided sufficiently densely, the discrete channel state information can be approximately equal to continuous channel state information. Therefore, the curve converges to the optimal delay-power tradeoff curve.

 More importantly, the delay-power tradeoff curve achieved by the scheduling based on discrete system state is piecewise linear. The policy corresponding to the vertex on this piecewise linear curve is the deterministic policy which is proofed in [11]. As we can see, when the state space of channel $\mathcal{H}$ is divided into 2 channel states, the distance between two adjacent vertices is shown to be 0.4944. When $\mathcal{H}$ is divided into 16 channel states, the distance decreases to 0.0753. It indicates that if $\mathcal{H}$ is divided more densely, the distance between two adjacent vertices will decrease. Therefore, when $\mathcal{H}$ is divided sufficiently densely, for any point on the delay-power tradeoff curve, we can always find a vertex on the curve to approximate it with sufficient small error, which exactly demonstrates the optimality of deterministic policies.
\section{Conclusion}
In this paper, we studied the optimal cross-layer scheduling based on the continuous channel state. More specifically, a variational problem was formulated to minimize average power for data transmission with the constraint of average delay. Based on this, we constructed a solution within the feasible region of the variational problem. The scheduling policy corresponding to the constructed solution was shown to be deterministic. Meanwhile, we show that the average power of the constructed solution can approximate that of the optimal solution with an arbitrarily small error. Therefore, the optimal scheduling policy could be found in the deterministic policies.

  \begin{figure}
   \centering
   \includegraphics[width=0.5\textwidth]{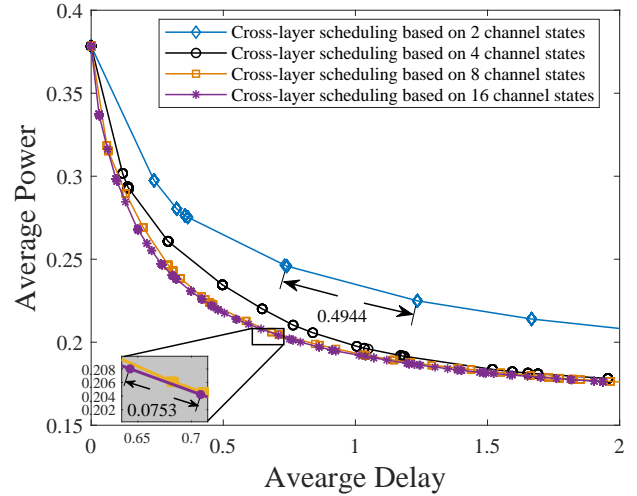}
   \caption{The delay-power tradeoff curves with discrete channel states.}
   \label{fig.pathdemo4}
   \end{figure}

\newpage

\end{document}